\documentclass[a4paper]{article}
\pdfoutput=1
\usepackage{amsmath}
\usepackage{amssymb}
\usepackage{amsthm}
\usepackage[hidelinks,breaklinks]{hyperref}

\usepackage{thmtools}
\usepackage{thm-restate}

\theoremstyle{plain}
\newtheorem{theorem}{Theorem}
\newtheorem{lemma}[theorem]{Lemma}
\newtheorem{corollary}[theorem]{Corollary}

\theoremstyle{definition}
\newtheorem{definition}[theorem]{Definition}
\newtheorem{example}[theorem]{Example}

\theoremstyle{remark}

\usepackage[svgnames]{xcolor}
\usepackage{tikz}
\usetikzlibrary{matrix,automata,arrows}
\usepackage{tikz-cd} 

\usepackage[electronic]{knowledge}
\knowledgeconfigure{notion,quotation,protect quotation={tikzcd}}

\usepackage{etoolbox}
\usepackage{microtype}%if unwanted, comment out or use option "draft"

\usepackage{wrapfig}

\usepackage{chngcntr}
\usepackage{apptools}
\AtAppendix{\counterwithin{theorem}{subsection}}

\usepackage{thmtools}
\usepackage{thm-restate}
\usepackage{cleveref}

\usepackage{cite}
\knowledge {knowledged}{url={https://ctan.org/pkg/knowledge?lang=en},italic}

\knowledgedirective{math}{notion}
\knowledgedirective{mathsymb}{math}

\knowledge{initial object}{}
\knowledge{final object}{}

\knowledge{duvs}{notion}
\knowledge{finite disjoint unions of vector spaces}{notion}
\knowledge{disjoint unions of vector spaces}{synonym}
\knowledge{disjoint union of vector spaces}{synonym}

\knowledge{index}[indices]{notion}

\knowledge{duvs map}[duvs maps]{notion}

\knowledge {duvs automaton}{notion}
\knowledge {duvs automata}{synonym}
\knowledge {automata over finite disjoint unions of vector spaces}{synonym}

\knowledge {variants}{notion}

\knowledge {hybrid set-vector automaton}{notion}
\knowledge {hybrid set-vector automata}{synonym}
\knowledge {hybrid automata}{synonym}
\knowledge {hybrid set-vector space automaton}{synonym}
\knowledge {hybrid set-vector space automata}{synonym}

\knowledge {$\MonoDiag \Vec $-automata}{notion}

\knowledge {minimal cover vs}{notion}
\knowledge {uniqueness cover vs}{notion}

\knowledge{diagram}[diagrams]{notion}

\knowledge{shape}[shapes]{notion}
\knowledge{indexing category}{notion}

\knowledge{cocone}[cocones]{}

\knowledge{structural components}[structural component]{}
\knowledge{presheaf}[presheaves]{notion}
\knowledge{category of elements}{notion}
\knowledge{Yoneda embedding}{notion}
\knowledge{petite}{notion}
\knowledge{representable presheaf}[representables|representable|representable presheaves]{notion}

\knowledge{factorization through}{notion}
\knowledge{factorization system through}[Factorization systems through]{synonym}
\knowledge {factorizing through a subcategory}{synonym}
\knowledge {factorization system through $\S $ on $\C$}{synonym}
\knowledge {factorization systems through}{synonym}
\knowledge {factorization system through~$\S $}{synonym}
\knowledge {factorization system through~$\categoryS $}{synonym}
\knowledge {factorization system through~$\categorySetFin $}{synonym}
\knowledge {factorization system through~$\categoryS $ on~$\categoryC $}{synonym}
\knowledge {factorization system through $\categoryComp $}{synonym}
\knowledge {($\EpiS $,$\MonoS $)-factorization of a category~$\categoryC $ through a subcategory~$\categoryS $}{synonym}
\knowledge {factorization system through $\GlueFinS $ on $\MonoDiagC $}{synonym}
\knowledge {factorization system through a subcategory~$\categoryS$}{synonym}
\knowledge {factorization system through $\categoryS $}{synonym}

\knowledge {(\Epi ,\Mono )-factorization system}{notion}
\knowledge {factorization systems through~$\categoryS $}{synonym}
\knowledge {(strong epi, mono) factorization system}{synonym}
\knowledge {factorization system on~$\categoryC$}{synonym}
\knowledge {factorization system on $\MonoDiagC $}{synonym}
\knowledge {$(\Epi,\Mono)$-factorization system}{synonym}
\knowledge {(split epi,mono)-factorization system}{synonym}
\knowledge {$(\Epi ,\Mono )$-factorization on~$\categoryC $}{synonym}

\knowledge {(\Epi ,\Mono )-factorization}{notion}
\knowledge {$(\EpiAutoS ,\MonoAutoS )$-factorization}{synonym}
\knowledge {$(\EpiS,\MonoS)$-factorization}{synonym}

\knowledge{diagonal}{notion}
\knowledge{diagonal fill}{synonym}

\knowledge {diagonal property}{notion}
\knowledge {diagonalization property}{synonym}
\knowledge {unique $(E_{\S },M_{\S })$-diagonalization property}{synonym}
\knowledge {unique diagonal property}{synonym}
\knowledge {unique diagonalization property}{synonym}
\knowledge {unique $(\EpiS ,\MonoS )$-diagonalization property}{synonym}

\knowledge {small morphism}{notion}
\knowledge {small}{synonym}
\knowledge {$\categoryS$-small}{synonym}
\knowledge {$\S$-small}{synonym}
\knowledge {-small}{synonym}
\knowledge {$\SetFin $-small morphisms}{synonym}
\knowledge {$\categoryVecFin $-small}{synonym}
\knowledge {$\S$-small arrows}{synonym}

\knowledge{automaton}[automata|Automata]{notion}
\knowledge {$\categoryS $-automaton}[$\categoryS $-automata]{synonym}
\knowledge {$\categoryC $-automaton}[$\categoryC $-automata]{synonym}
\knowledge {$\C $-automaton}[$\C $-automata]{synonym}
\knowledge {$\S $-automaton}[$\S $-automata]{synonym}
\knowledge {automaton $\automatonA $ in the category~$\categoryC $}{synonym}
\knowledge {$\categoryC ,I,F$-automata}{synonym}
\knowledge {(finite word) automaton in a category}{synonym}
\knowledge {$\categoryC ,I,F$-automaton}{synonym}
\knowledge {Automata in a category}{synonym}
\knowledge {automaton in a category}{synonym}

\knowledge {state object}[state objects]{notion}

\knowledge {initial morphism}[initial morphisms]{notion}

\knowledge {final morphism}[final morphisms]{notion}
  
\knowledge {category of~$\categoryC ,I,F$-automata for $L$}{notion}

\knowledge {automata morphism}[automata morphisms|Automata morphisms]{notion}
\knowledge {morphisms of $\categoryC ,I,F$-automata}{synonym}
\knowledge {$\categoryCautomataL$-morphisms}{synonym}
\knowledge {$\categorySautomataL$-morphism}{synonym}
\knowledge {$\categoryS$-automaton morphism}{synonym}
\knowledge {morphism of $\categoryC,I,F$-automata}{synonym}
\knowledge {$\categoryCautomataL$-morphism}{synonym}
\knowledge {$\categoryC $-automaton morphism}{synonym}
\knowledge {$\categoryC $-automata morphism}{synonym}

\knowledge {language}[languages|Languages]{notion}
\knowledge {$\categoryC $-language}{synonym}
\knowledge {$\categoryC$-languages}{synonym}
\knowledge{$\categoryC ,I,F$-language}{synonym}
\knowledge {$\GlueC $-languages}{synonym}

\knowledge {accepts the language~$L$}{notion}
\knowledge {accepts the language $L$}{synonym}
\knowledge {$\categoryS $-automaton for~$L$}{synonym}
\knowledge {accept the language~$L$}{synonym}
\knowledge {accepts}{synonym}
\knowledge {accepted}{synonym}
\knowledge {$\categoryS$-automata~$\automatonA$ for~$L$}{synonym}

\knowledge{final automaton}{notion}
\knowledge {final $\categoryC $-automaton}{synonym}
\knowledge{final automata for a language}{synonym}

\knowledge {initial automaton}{notion}
\knowledge {initial and final automata for a language}{synonym}
  
\knowledge{minimal auotmaton}{notion}
\knowledge {minimal $\categoryC $-automaton for~$L$}{synonym}

\knowledge{minimal assumptions}{notion}

\knowledge{vector space automaton}[vector space automata|Vector space automata]{notion}
\knowledge {$\Vec$-automaton}[$\Vec$-automata]{synonym}
\knowledge {$\VecFin$-automaton}[$\VecFin$-automata]{synonym}
\knowledge {finite automata weighted over a field}{synonym}

\knowledge {deterministic finite state automata}{notion}
\knowledge {$\categorySetFin $-automata}{synonym}
\knowledge {$\categorySet $-automata}{synonym}
\knowledge {set automaton}{synonym}
\knowledge {deterministic automaton}{synonym}

\knowledge {set of states}{notion}
\knowledge {set of accepting states}{notion}
\knowledge {initial state}{notion}

\knowledge {free colimit completion of $\C$}{notion}

\knowledge{mono-cocone}[\Mono-cocone|\Mono-cocones]{notion}
\knowledge {$\Mono $-cocone}{synonym}

\knowledge{\Mono-diagram}[\Mono-diagrams]{notion}
\knowledge{$\Mono$-diagram}[$\Mono$-diagrams]{synonym}

\knowledge{glued category}{link=\Glue}

\knowledge{glued vector spaces}{notion}

\knowledge {diagram morphism}[diag morphism]{notion}

\knowledge {$\categoryS $-extremal epimorphism}{notion}
\knowledge {$\categoryS $-extremal epimorphisms}{synonym}
\knowledge {$\categoryS $-extremal morphisms}{synonym}
\knowledge{$\S$-extremal epimorphism}[$\S$-extremal epimorphisms]{synonym}

\knowledge {intersection}[intersections]{}

\knowledge{$\S$-subobject}[$\S$-subobjects]{notion}
\knowledge {$\MonoS $-subobjects}[$\MonoS$-subobject]{synonym}
\knowledge {$\MonoS$-subobject in $\S$}{synonym}
\knowledge {$\MonoMDC $-subobjects}{synonym}
\knowledge {$\MonoAutoS $-subobject}{synonym}

\knowledge {proper subobject}{notion}
\knowledge {proper}{synonym}

\knowledge{\preshC}{mathsymb}
\knowledge{\DiagC}{mathsymb}
\knowledge{\MonoDiagC}{mathsymb}
\knowledge{\Set}{mathsymb}
\knowledge{\El}{mathsymb}

\knowledge{\X}{mathsymb}

\knowledge{\Y}{mathsymb}
\knowledge{\Fhat}{mathsymb}
\knowledge{\Ghat}{mathsymb}
\knowledge{\repr}{mathsymb}

\knowledge {\initialInitialStateL }{}
\knowledge {\initialFinalMapL }{}
\knowledge {\initialTransitionMapL }{}
\knowledge {\initialStateObjectL }{}

\knowledge {\finalInitialStateL }{}
\knowledge {\finalTransitionMapL }{}
\knowledge {\finalFinalMapL }{}
\knowledge {\finalStateObjectL }{}

\knowledge {$\EpiAutoS $-quotient}{}
\knowledge {$\EpiMDC $-subobjects}{}

\knowledge {leading class}{notion}

\knowledge {covered}{notion}

\knowledge {rank main}{notion}

\knowledge {rank}{notion}
\knowledge {ranks~$0,1,\dots $}{synonym}
\knowledge {rank~$i$}{synonym}
\knowledge {$i$'th rank}{synonym}
\knowledge {rank~$j$}{synonym}

\knowledge{Global homogeneity}[global homogeneity]{notion}

\knowledge {well-behaved}{notion}

\newenvironment{dani-comm}{
 \color{red}
  { D's comment:
  }
  }

\knowledge\C{}
\knowledge\D{}
\knowledge\E{}
\knowledge\S{}
\knowledge{\category C}{}
\knowledge{\category D}{}
\knowledge{\category E}{}
\knowledge{\category S}{}
\knowledge {\functor F}{}
\knowledge {\functor G}{}
\knowledge {\functor H}{}
\knowledge {\functor K}{}
\knowledge {\functor LF}{}

\knowledge {\EpiS }{}
\knowledge {\MonoS }{}
\knowledge {\Mono }{}
\knowledge {\Epi }{}
\knowledge {\EpiAutoS }{}
\knowledge {\MonoAutoS }{}
\knowledge {\functor F_0}{}
\knowledge {\functor F_1}{}
\knowledge {\functor F_2}{}

\newrobustcmd{\C}{\kl[\C]{\mathcal{C}}} % generic small cat
\newrobustcmd{\D}{\kl[\D]{\mathcal{D}}} % generic small cat
\newrobustcmd{\E}{\kl[\E]{\mathcal{E}}} % generic small cat
\newrobustcmd{\preshC}{\kl[\preshC]{[\mathcal{C}^\op,\mathsf{Set}]}} % presheaf cat over \C.

\newrobustcmd{\El}[1]{\kl[\El]{\mathit{El}({#1})}}
\newrobustcmd{\Elmono}[1]{\kl[\Elmono]{\mathit{El}_{\mathit{Mono}}({#1})}}
\renewrobustcmd{\S}{\kl[\S]{\mathcal{S}}}

\knowledge\Diag{math}
\knowledge\DiagFin{math}

\newrobustcmd\Diag[1]{\kl[\Diag]{\mathsf{Diag}}(#1)}
\newrobustcmd\DiagFin[1]{\kl[\DiagFin]{\mathsf{DiagFin}_\fin}(#1)}

\newrobustcmd\DiagC{\Diag{\categoryC}} 
\newrobustcmd\DiagS{\Diag\categoryS}
\newrobustcmd\DiagFinS{\DiagFin\categoryS}

\knowledge\Glue{math}
\knowledge\GlueFin{math}

\newrobustcmd\Glue{\kl[\Glue]{\mathtt{Glue}}}
\newrobustcmd\GlueFin{\kl[\GlueFin]{\mathtt{Glue}_{\fin}}}

\newrobustcmd\GlueC{\Glue(\categoryC)}
\newrobustcmd\GlueS{\Glue(\categoryS)}
\newrobustcmd\GlueFinS{\GlueFin(\categoryS)}

\newrobustcmd\MonoDiag[1]{\Glue(#1)}
\newrobustcmd\MonoDiagFin[1]{\GlueFin(#1)}
\let\MonoDiagC\GlueC

\let\MonoDiagFinS\GlueFinS

\newrobustcmd{\X}{\kl[\X]{X}} % generic object of a cat

\newrobustcmd{\Y}{\kl[\Y]{\mathcal{Y}}} % Yoneda embedding
\newrobustcmd{\Fhat}{\kl[\Fhat]{\widehat{F}}} % colimit of a \C
\newrobustcmd{\Ghat}{\kl[\Ghat]{\widehat{\functorG}}} % colimit of a \C
\newrobustcmd\moustache[1]{\kl[\Ghat]{\widehat{#1}}}

\newrobustcmd{\repr}[1]{\kl[\repr]{\mathcal{C}(-,{#1})}} % representable
\newrobustcmd{\op}{\mathit{op}}
\newrobustcmd{\colim}{\mathsf{colim}}

\newrobustcmd\Reals{\mathbb{R}}

\knowledge{moustache}{math}
\newrobustcmd\equivf[1]{\mathrel{\kl[moustache]{\sim_{#1}}}}
\newrobustcmd\equivG{\mathrel{\kl[moustache]{\sim_{\functorG}}}}
\newrobustcmd\equivF{\mathrel{\kl[moustache]{\sim_{\functorF}}}}
\newrobustcmd\equivK{\mathrel{\kl[moustache]{\sim_{\functorK}}}}
\newrobustcmd\equivH{\mathrel{\kl[moustache]{\sim_{\functorH}}}}

\newrobustcmd\colimG{\colim{\functorG}}

\newrobustcmd\fieldK{\mathbb{K}}

\knowledge\exVectorMap{math}
\newcommand\exVectorMap{\kl[\exVectorMap]{F}}

\knowledge\automatonVec{math}
\newcommand\vectorial{{\kl[\automatonVec]{\mathtt{vec}}}}
\newrobustcmd\automatonVec{\automatonA^\vectorial}
\newrobustcmd\statesVec{Q^\vectorial}
\newrobustcmd\finalMapVec{f^\vectorial}
\newrobustcmd\initialStateVec{i^\vectorial}
\newrobustcmd\transitionMapVec{\delta^\vectorial}

\knowledge{\automatonDuvs}{math}
\newcommand\duvs{{\kl[\automatonDuvs]{\mathtt{duvs}}}}
\newrobustcmd\automatonDuvs{\automatonA^\duvs}
\newrobustcmd\statesDuvs{Q^\duvs}
\newrobustcmd\finalMapDuvs{f^\duvs}
\newrobustcmd\initialStateDuvs{i^\duvs}
\newrobustcmd\transitionMapDuvs{\delta^\duvs}

\knowledge\autosem{math}
\newrobustcmd\autosem[1]{\kl[\autosem]{[\![}#1\kl[\autosem]{]\!]}}
\newrobustcmd\alphabet[1][A]{\mathtt{#1}}
\newrobustcmd\alphabetA{\alphabet[A]}

\newrobustcmd\automaton[1]{\mathcal{#1}}%{\kl[$\C$-automaton]{\mathcal{#1}}}

\newrobustcmd\automatonA{\automaton A}
\newrobustcmd\automatonB{\automaton B}
\newrobustcmd\automatonC{\automaton C}

\knowledge{\stateObject}{mathsymb}
\knowledge{\initialState}{mathsymb}
\knowledge{\finalMap}{mathsymb}

\newrobustcmd\stateObject{\kl[\stateObject]{Q}}
\newrobustcmd\initialState{\kl[\initialState]{i}}
\newrobustcmd\finalMap{\kl[\finalMap]{f}}
\newrobustcmd\transitionMap{\delta}

\newrobustcmd\initialStateObjectL{\kl[\initialStateObjectL]{Q^{\texttt{init}(L)}}}
\newrobustcmd\initialInitialStateL{\kl[\initialInitialStateL]{i^{\texttt{init}(L)}}}
\newrobustcmd\initialFinalMapL{\kl[\initialFinalMapL]{f^{\texttt{init}(L)}}}
\newrobustcmd\initialTransitionMapL{\kl[\initialTransitionMapL]\delta^{\kl[\initialTransitionMapL]{\texttt{init}(L)}}}

\newrobustcmd\finalStateObjectL{\kl[\finalStateObjectL]{Q^{\texttt{final}(L)}}}
\newrobustcmd\finalInitialStateL{\kl[\finalInitialStateL]{i^{\texttt{final}(L)}}}
\newrobustcmd\finalFinalMapL{\kl[\finalFinalMapL]{f^{\texttt{final}(L)}}}
\newrobustcmd\finalTransitionMapL{\kl[\finalTransitionMapL]\delta^{\kl[\finalTransitionMapL]{\texttt{final}(L)}}}

\newrobustcmd\morphismAutomataH{h}

\knowledge\initialCautomaton{math}
\knowledge\finalCautomaton{math}

\newrobustcmd\initialCautomaton{\kl[\initialCautomaton]{\mathtt{init}_{\categoryC}}}
\newrobustcmd\finalCautomaton{\kl[\finalCautomaton]{\mathtt{final}_{\categoryC}}}
\newrobustcmd\initialCautomatonL{\initialCautomaton(L)}
\newrobustcmd\finalCautomatonL{\finalCautomaton(L)}

\knowledge\categoryAutomataL{math}
\knowledge {\categoryCautomataL }{synonym}
\newrobustcmd\categoryCautomataL{\kl[\categoryAutomataL]{\mathsf{Auto}_{\categoryC}(L)}}
\newrobustcmd\categorySautomataL{\kl[\categoryAutomataL]{\mathsf{Auto}_{\categoryS}(L)}}
\newrobustcmd\categoryCautomataLOp{\kl[\categoryAutomataL]{\mathsf{Auto}_{\categoryCOp}(L^\op)}}

\knowledge\reachSautomaton{math}
\knowledge\obsSautomaton{math}

\newrobustcmd\reachSautomaton{\kl[\reachSautomaton]{\mathtt{reach}_{\categoryS}}}
\newrobustcmd\reachSautomatonA{\reachSautomaton(\automatonA)}
\newrobustcmd\obsSautomaton{\kl[\obsSautomaton]{\mathtt{obs}_{\categoryS}}}
\newrobustcmd\obsSautomatonA{\obsSautomaton(\automatonA)}

\knowledge\minimalSautomaton{math}
\newrobustcmd\minimalSautomaton{\kl[\minimalSautomaton]{\mathtt{min}_{\categoryS}}}
\newrobustcmd\minimalSautomatonL{\minimalSautomaton(L)}

\let\reachAutomatonA\reachSautomatonA

\newrobustcmd\obsReachAutomatonA{\obsAutomaton(\reachAutomaton(\automatonA))}
\let\minAutomatonL\minimalSautomatonL

\NewDocumentCommand\rightmonoarrow{o}{\IfNoValueTF{#1}
     {\rightarrowtail}
     {\stackrel{#1}{\rightarrowtail}}}

\newcommand\fin{{\mathsf{fin}}}
\knowledge\categorySet{}
\knowledge\categorySetFin{}
\knowledge\categoryVec{}
\knowledge\categoryVecFin{}
\knowledge\categoryTop{}
\knowledge\categoryComp{}

\newrobustcmd\categorySet{\kl[\categorySet]{\mathsf{Set}}} % cat of sets and
\let\Set\categorySet
\newrobustcmd\categorySetFin{\kl[\categorySetFin]{\mathsf{Set}_\fin}} % cat of sets and
\let\SetFin\categorySetFin
\newrobustcmd\categoryVec{\kl[\categoryVec]{\mathsf{Vec}}} % cat of sets and
\let\Vec\categoryVec
\newrobustcmd\categoryVecFin{\kl[\categoryVecFin]{\mathsf{Vec}_\fin}} % cat of sets and
\let\VecFin\categoryVecFin
\newrobustcmd\categoryTop{\kl[\categoryTop]{\mathsf{Top}}}
\newrobustcmd\categoryComp{{\kl[\categoryTop]{\mathsf{Comp}}}}

\newrobustcmd\category[1]{\kl[\category#1]{\mathcal{#1}}}
\newrobustcmd\categoryC{\category C}
\newrobustcmd\categoryCOp{\category C^\op}
\newrobustcmd\categoryD{\category D}
\newrobustcmd\categoryE{\category E}
\newrobustcmd\categoryS{\category S}

\newrobustcmd\functor[1]{\kl[\functor#1]{\mathtt{#1}}}
\newrobustcmd\functorF{\functor F}
\newrobustcmd\functorG{\functor G}
\newrobustcmd\functorH{\functor H}
\newrobustcmd\functorK{\functor K}
\newrobustcmd\functorLF{\functor{LF}}

\newrobustcmd\functorCategoryDS{[\categoryD,\categoryS]}
\newrobustcmd\functorCategoryDC{[\categoryD,\categoryC]}

\newrobustcmd\Epi{\ensuremath{\kl[\Epi]{E}}}
\newrobustcmd\EpiS{\ensuremath{\kl[\EpiS]{E_{\categoryS}}}}
\newrobustcmd\Mono{\ensuremath{\kl[\Mono]M}}
\newrobustcmd\MonoS{\ensuremath{\kl[\MonoS]{M_{\categoryS}}}}

\newrobustcmd\EpiMDC{\ensuremath{\kl[\EpiMDC]{\mathit{Epi}_\mathit{MDC}}}}
\newrobustcmd\MonoMDC{\ensuremath{\kl[\MonoMDC]{\mathit{Mono}_\mathit{MDC}}}}

\newrobustcmd\EpiDS{\ensuremath{\kl[\EpiDS]{E_{[\categoryD,\categoryS]}}}}
\newrobustcmd\MonoDS{\ensuremath{\kl[\MonoDS]{M_{[\categoryD,\categoryS]}}}}

\newrobustcmd\EpiAutoS{\ensuremath{\kl[\EpiAutoS]{E_{\categorySautomataL}}}}
\newrobustcmd\MonoAutoS{\ensuremath{\kl[\MonoAutoS]{M_{\categorySautomataL}}}}

\knowledge\EpiGlueC{math}
\knowledge\MonoGlueC{math}
\newrobustcmd\EpiGlueC{\ensuremath{\kl[\EpiGlueC]{\mathit{Epi}_{\GlueC}}}}
\newrobustcmd\MonoGlueC{\ensuremath{\kl[\MonoGlueC]{\mathit{Mono}_{\GlueC}}}}
\let\EpiMDC\EpiGlueC
\let\MonoMDC\MonoGlueC

\newrobustcmd\id{\mathit{id}}

\title{Automata in the Category of Glued Vector Spaces\footnote{This
    work was supported by the European Research Council (ERC) under
    the European Union’s Horizon 2020 research and innovation
    programme (grant agreement No.670624), and by the DeLTA ANR
    project (ANR-16-CE40-0007). The authors also thank the Simons
    Institute for the Theory of Computing where this work has been
    partly developped.} \footnote{This is the "knowledge@knowledged" 
    enriched version of the same paper published in the proceedings of MFCS 2017.}} 
\author{%
	\begin{tabular}{c}%
	Thomas Colcombet\\
	\texttt{thomas.colcombet@irif.fr}%
	\end{tabular}
	\begin{tabular}{c}%
	Daniela Petri{\c s}an\\
	\texttt{daniela.petrisan@irif.fr}%
	\end{tabular}\\[0.6cm]
	CNRS, IRIF, Univ. Paris-Diderot, Paris 7, France}

\begin{document}

\maketitle

\begin{abstract}
  In this paper we adopt a category-theoretic approach to the
  conception of automata classes enjoying minimization by design. The
  main instantiation of our construction is a new class of automata
  that are hybrid between deterministic automata and automata weighted
  over a field.
\end{abstract}

\section{Introduction}
\label{section:introduction}

In this paper we introduce a new automata model,\emph{
  "hybrid set-vector automata"}, designed to accept weighted languages
over a field in a more efficient way than Schützenberger's weighted
automata~\cite{schutzenberger61}. The space of states for these
automata is not a vector space, but rather a union of vector spaces 
``glued'' together along subspaces.  We call them "hybrid automata",
since they naturally embed both \kl{deterministic finite state
  automata} and \kl{finite automata weighted over a field}. In
Section~\ref{section:informal presentation} we present at an informal
level a motivating example and the intuitions behind this
construction, avoiding as much as possible category-theoretical
technicalities. We use this example to guide us throughout the rest of
the paper.

A key property that the new automata model should satisfy is
minimization. Since the morphisms of ``glued'' vector spaces are
rather complicated to describe, proving the existence of minimal
automata ``by hand'' is rather complicated. Therefore we opted for a
more systematic approach and adopted a category-theoretic perspective
for designing \emph{new forms of automata} that enjoy
\emph{minimization by design}. In particular, we introduce the
category of ``glued'' vector spaces in which these automata should
live and we analyse its properties that render minimization possible.

Starting with the seminal papers of Arbib and Manes, see for
example~\cite{ArbibManes75} and the references therein, and of
Goguen~\cite{goguen1972}, it became well established that category
theory offers a neat understanding of several phenomena in automata
theory. In particular, the key property of minimization in different
contexts, such as for deterministic automata (over finite words) and
Schützenberger's automata weighted over fields
\cite{schutzenberger61}, arises from the same categorical reasons
(existence of some limits/colimits and an (epi,mono)-factorization
system \cite{ArbibManes75}).

There is a long tradition of seeing automata either as algebras or
coalgebras for a functor. However, in the case of deterministic
automata, the algebraic view does not capture the accepting states,
while the coalgebraic view does not capture the initial state. In the
coalgebraic setting one needs to consider the so-called pointed
coalgebras, see for example~\cite{AdamekEtAl:Well-Pointed-CoAlg},
where minimal automata are modelled as well-pointed coalgebras.  The
dual perspective of automata seen as both algebras or coalgebras, as
well as the duality between reachability and observability, has been
explored more recently in papers such
as~\cite{BonchiBHPRS14,Bonchi2012,Bezhanishvili2012}.

Here we take yet another approach to defining automata in a
category. The reader acquainted with category theory will recognise
that we see automata as functors from an input category (that
specifies the type of the machines under consideration, which in this
paper is restricted to word automata) to a category of output values.
We show that the next ingredients are sufficient to ensure
minimization: the existence of an initial and of a final automaton for
a language, and a factorization system on the category in which we
interpret our automata.

For example, deterministic and weighted automata over a field are
obtained by considering as output categories the categories $\Set$ of
sets and functions and $\Vec$ of vector spaces and linear maps,
respectively. Since $\Set$ and $\Vec$ have all limits and colimits, it
is very easy to prove the existence of initial and final automata
accepting a given language. In both cases, the minimal automaton for a
language is obtained by taking an epi-mono factorization of the unique
arrow from the initial to the final automaton.

Notice that the initial and the final automata have infinite
(-dimensional) state sets (spaces). If the language at issue is
regular, that is, if the unique map from the initial to the final
automaton factors through a finite (-dimensional) automaton then,
automatically, the minimal automaton will also be finite
(-dimensional).  However, this relies on very specific properties of
the categories of sets and vector spaces, namely on the fact that the
full subcategories $\categorySetFin$ of finite sets and
$\categoryVecFin$ of finite dimensional vector spaces are closed in
$\categorySet$, respectively in $\categoryVec$, under both quotients
and subobjects.

Coming back to "hybrid set-vector automata", we define them as word
automata interpreted in an output category $\MonoDiag\Vec$ which we
obtain as the completion of~$\categoryVec$ under certain colimits, and
can be described at an informal level as ``glueings'' of arbitrary
vector spaces. The definition of this form of cocompletion
$\MonoDiagC$ of a category $\categoryC$ is the subject of
Section~\ref{section:gluing of categories}.

We are interested in those "hybrid automata" for which the state
object admits a finitary description, which intuitively can be
described as \emph{finite} glueings of \emph{finite dimensional}
spaces. For this reason we will consider the subcategory
$\GlueFin(\VecFin)$ of $\MonoDiag\Vec$. It turns out that
$\GlueFin(\VecFin)$ is closed under quotients in $\MonoDiag\Vec$ but,
crucially, \emph{it is not closed under subobjects}. For example, a
glueing of infinitely many one-dimensional spaces is a subobject of a
two-dimensional space, but only the latter is an object of
$\GlueFin(\VecFin)$.

This is the motivation for introducing a notion of
\reintro{($\EpiS$,$\MonoS$)-factorization of a category~$\categoryC$
  through a subcategory~$\categoryS$}. This is a refinement of the
classical notion of factorization system on $\categoryC$ and is used
for isolating the semantical computations (in~$\categoryC$) from the
automata themselves (with an object from~$\categoryS$ as ``set of
configurations'').\footnote{This distinction is usually not necessary,
  and we are not aware of its existence in the literature. It is
  crucial for us, thus we cannot use already existing results from the
  coalgebraic literature, e.g.~\cite{AdamekEtAl:Well-Pointed-CoAlg}.}
We show how it provides a minimization of ``\kl{$\categoryS$-automata}
for representing \kl{$\categoryC$-languages}''.  A concrete instance
of this is a factorization system on $\MonoDiag\Vec$ through
$\MonoDiagFin\VecFin$, which plays a crucial role in proving the
existence of minimal $\MonoDiagFin\VecFin$-automata for recognizing
weighted languages.

The rest of the paper is organised as follows. We first develop a
motivating example of a \kl{hybrid set-vector space automaton} in
Section~\ref{section:informal presentation}. We then identify in
Section~\ref{section:categorical automata} the category-theoretic
ingredients that are sufficient for a class of automata to enjoy
minimization. We then turn to our main contribution, namely the
description and the study of the properties of (finite-)mono-diagrams
in a category, in Section~\ref{section:gluing of categories}.  We
conclude in Section~\ref{section:conclusion} with a discussion of some
of the design choices we made in this paper.

\section{The hybridisation of deterministic finite state and vector
  automata}
\label{section:informal presentation}
\newcommand\mysubparagraph[1]{\smallskip\noindent\textbf{#1}}
\let\subparagraph\mysubparagraph

In this section, we (rather informally) describe the motivating
example of this paper: the construction of a family of automata that
naturally extends both \kl{deterministic finite state automata} and
\kl{finite automata weighted over a field} in the sense of
Schützenberger (i.e., automata in the category of finite vector
spaces).  The intuition should then support the categorical
constructions that we develop in the subsequent sections.

\subparagraph{Set automata (deterministic automata).} \AP Let us fix
ourselves an alphabet~$\alphabetA$.
A ""deterministic automaton""
(or \reintro{set automaton}) is a tuple
\begin{align*}
  \automatonA=(\stateObject,\initialState, \finalMap,
  (\transitionMap_a)_{a\in\alphabet})\,,
\end{align*}
in which $\stateObject$ is a \intro{set of states}, $\initialState$ is
a map from a one element set~$1=\{0\}$ to~$\stateObject$ (i.e. an
\intro{initial state}), $\finalMap$ is a map from~$\stateObject$ to a
two elements set~$2=\{0,1\}$ (i.e. a \intro{set of accepting states}),
and~$\transitionMap_a$ is a map from~$\stateObject$ to itself for all
letters~$a\in\alphabet$. \AP Given a word~$u=a_1\ldots a_n$, the
automaton \reintro{accepts} the
map~$\autosem\automatonA(u)\colon1\rightarrow2$ defined as:
\begin{align*}
  \autosem\automatonA(u)&=\finalMap\circ\transitionMap_u\circ\initialState
  &\text{where}\quad\transitionMap_{a_1\dots
    a_n}&=\transitionMap_{a_n}\circ\cdots\circ\transitionMap_{a_1}\,.
\end{align*}
We recognize here the standard definition of a deterministic
automaton, in which a word~$u$ is accepted if the
map~$\autosem\automatonA(u)$ is the constant~$1$, and rejected if it
is the constant~$0$.

\subparagraph{Vector space automata (automata weighted over a field).}
\AP
Now, we can use the same definition of an automaton, this time
with~$\stateObject$ a vector space (over, say, the field $\Reals$),
$\initialState$ a linear map from $\Reals$ to~$\stateObject$,
$\finalMap$ a linear map from~$\stateObject$ to~$\Reals$ (seen as a
$\Reals$-vector space as usual), and $\transitionMap_a$ a linear map
from~$\stateObject$ to itself. In other words, we have used the same
definition, but this time in the category of vector spaces. Given a
word~$u$, a \kl{vector space automaton} $\automatonA$ computes
$\autosem\automatonA(u)\colon \Reals\rightarrow \Reals$ as the
composite described above. Since a linear map from~$\Reals$
to~$\Reals$ is only determined by the image of~$1$, this automaton can
be understood as associating to each input word~$u$ the real
number~$\autosem\automatonA(u)(1)$.  We will informally refer to such
automata in this section as \intro{vector space automata}. Let us
provide an example.

\subparagraph{Leading example.} \AP For a word $u\in \{a,b,c\}^*$ let
$|u|_a$ denote the number of occurrences of the letter $a$ in $u$.
Let us compute the map~$\intro*\exVectorMap$ which, given a word~$u\in
\{a,b,c\}^*$, outputs~$2^{|u|_a}$ if it contains an even number
of~$b$'s and no~$c$'s, and $0$ in all other cases. This is achieved
with the \kl{vector space automaton}
\AP\intro[\automatonVec]{}$\automatonVec=(\statesVec=\Reals^2,\initialStateVec,\finalMapVec,\transitionMapVec)$
where for all~$x,y\in\Reals$,
\begin{align*}
  \initialStateVec(x)&=(x,0)\,, & \transitionMapVec_a(x,y)&=(2x,2y)\,,& \transitionMapVec_b(x,y)&=(y,x)\,,\\
  \finalMapVec(x,y)&=x\,,& \transitionMapVec_c(x,y)&=(0,0)\,.
\end{align*}
One easily checks that
indeed~$\autosem\automatonVec(u)(1)=\exVectorMap(u)$ for all
words~$u\in\alphabetA^*$.

\subparagraph{Can we do better?} \AP It is well known from
Schützenberger's seminal work that the \kl{vector space
  automaton}~$\automatonVec$ is minimal, both in an algebraic sense
(to be described later) as well as at an intuitive level in the sense
that no \kl{vector space automaton} could recognize $\exVectorMap$
with a dimension one vector space as configuration space:
$\automatonVec$ is ``dimension minimal.''

However, let us think for one moment on how one would ``implement''
the function~$\exVectorMap$ as an online device that would get letters
as input, and would modify its internal state accordingly. Would we
implement concretely~$\automatonVec$ directly? Probably not, since
there is a more economic\footnote{Under the reasonable assumption that
  maintaining a real is more costly than maintaining a bit.} way to
obtain the same result: we can maintain $2^m$ where~$m$ is the number
of~$a$'s seen so far, together with one bit for remembering whether
the number of~$b$'s is even or odd.  Such an automaton would start
with~$1$ in its unique real valued register. Each time an~$a$ is met,
the register is doubled, each time~$b$ is met, the bit is reversed,
and when~$c$ is met, the register is set to~$0$. At the end of the
input word, the automaton would output~$0$ or the value of the
register depending on the current value of the bit. If we consider the
configuration space that we use in this encoding, we use
$\Reals\uplus\Reals$ instead of~$\Reals\times\Reals$. Can we define an
automata model that would faithfully implement this example?

\subparagraph{A first generalization: \kl{disjoint unions of vector
    spaces}.} \AP A way to achieve this is to interpret the generic
notion of \kl{automata} in the category of \kl{finite disjoint unions
  of vector spaces} (\intro{duvs}).  One way to define such a
\intro{finite disjoint unions of vector spaces} is to use a finite
set~$N$ of `\intro{indices}'~$p,q,r\dots$, and to each \kl{index}~$p$
associate a vector space~$V_p$. The `space' represented is then
$\{(p,\vec v)\mid p\in N,\,\vec v\in V_p\}$.  A `map' between
\kl{duvs} represented by~$(N,V)$ and~$(N',V')$ is then a
pair~$h:N\rightarrow N'$ together with a linear map~$f_p$ from~$V_p$
to~$V'_{h(p)}$ for all~$p\in N$. It can be seen as mapping each
$(p,\vec v)\in N\times V_p$ to $(h(p),f_p(\vec v))$. Call this a
\intro{duvs map}.  Such \kl{duvs maps} are composed in a natural way.
This defines a category, and hence we can consider \intro{duvs
  automata} which are automata with a \kl{duvs} for its state space,
and transitions implemented by \kl{duvs maps}.

\AP For instance, we can pursue with the computation of~$\exVectorMap$
and provide a \kl{duvs automaton}
\intro[\automatonDuvs]{}$\automatonDuvs=(\statesDuvs,\initialStateDuvs,\finalMapDuvs,\transitionMapDuvs)$
where $\statesDuvs=\{(s,x)\mid s\in\{\mathtt{even},\mathtt{odd}\},\
x\in\Reals\}$ (considered as a \kl{disjoint union of vector spaces}
with \kl{indices} $\mathtt{even}$ and $\mathtt{odd}$ and all
associated vector spaces~$V_\mathtt{even}=V_\mathtt{odd}=\Reals$).
The maps can be conveniently defined as follows:
\begin{align*}
  \initialStateDuvs(x)&=(\mathtt{even},x)&
  \transitionMapDuvs_a(\mathtt{even},x)&=(\mathtt{even},2x)&
  \transitionMapDuvs_a(\mathtt{odd},x)&=(\mathtt{odd},2x)\\
  \finalMapDuvs(\mathtt{even},x)&=x&
  \transitionMapDuvs_b(\mathtt{even},x)&=(\mathtt{odd},x)&
  \transitionMapDuvs_b(\mathtt{odd},x)&=(\mathtt{even},x)\\
  \finalMapDuvs(\mathtt{odd},x)&=0&
  \transitionMapDuvs_c(\mathtt{even},x)&=(\mathtt{even},0)&
  \transitionMapDuvs_c(\mathtt{odd},x)&=(\mathtt{odd},0)
\end{align*}
This automaton computes the expected~$\exVectorMap$. It is also
obvious that such \kl{automata over finite disjoint unions of vector
  spaces} generalize both \kl{deterministic finite state automata}
(using only $0$-dimensional vector spaces), and \kl{vector space
  automata} (using only one \kl{index}).  However, is it the joint
generalization that we hoped for? The answer is no...

\subparagraph{Minimization of "duvs" automata.} \AP We could think
that the above automaton~$\automatonDuvs$ is minimal. However, it
involved some arbitrary decisions when defining it. This can be seen
in the fact that when~$\transitionMapDuvs_c$ is applied, we chose to
not change the index (and set to null the real value): this is
arbitrary, and we could have exchanged $\mathtt{even}$ and
$\mathtt{odd}$, or fixed it abitrarily to $\mathtt{even}$, or to
$\mathtt{odd}$. All these \intro{variants} would be equally valid for
computing~$\exVectorMap$.

\AP It is a bit difficult at this stage to explain the non-minimality
of these automata since we did not introduce the proper notions
yet. Let us try at a high level, invoking some standard
automata-theoretic concepts.  The first remark is that every
configuration in $\statesDuvs$ is `reachable' in this automaton:
indeed $(\mathtt{even},x)=\initialStateDuvs(x)$ and
$(\mathtt{odd},x)=\transitionMapDuvs_b\circ\initialStateDuvs(x)$ for
all~$x\in\Reals$. Hence there is no hope to improve the automaton
$\automatonDuvs$ or one of its \kl{variants} by some form of
`restriction to its reachable configurations'. Only `quotienting of
configurations' remains. However, one can show that none among
$\automatonDuvs$ and the variants mentioned above is the quotient of
another. Keeping in mind the Myhill-Nerode equivalence, we should
instead merge the configurations $(\mathtt{even},0)$ and
$(\mathtt{odd},0)$ since these are observationally equivalent:
\begin{align*}
  \finalMapDuvs\circ\transitionMapDuvs_u(\mathtt{even},0)&= 0 =
  \finalMapDuvs\circ\transitionMapDuvs_u(\mathtt{odd},0)\ &&\text{for
    all words~$u\in\alphabetA^*$.}
\end{align*}
However, the quotient "duvs" obtained by merging $(\mathtt{even},0)$
and $(\mathtt{odd},0)$, albeit not very intuitive, consists of one
index associated to a two dimensional vector space, which is
essentially an indexed version of the "vector space automaton"
$\automatonVec$ computed before.  At this stage, we understand that
minimising in the category of "duvs" is not very helpful, as we do not
obtain the desired optimisation.

\subparagraph{How to proceed from here.}  The only reasonable thing to
do is indeed to merge~$(\mathtt{even},0)$ and $(\mathtt{odd},0)$, but
we have to be more careful about the precise meaning of `quotient'.  A
possibility is to add explicitly equivalence classes in the definition
of the automaton.  However, category theory provides useful concepts
and terminology for defining these objects: colimits, and more
precisely the free co-completion of a category. In the previous
paragraph, we have shown that the category of "duvs" -- which is
itself the free completion of $\categoryVec$ with respect to finite
coproducts -- is not a good ambient category for our purposes. We need
more colimits, so that the notion of `quotient' is further refined. At
the other extreme, we could consider the free completion with respect
to all colimits, which, informally, consists of objects obtained
from the category using copying and gluing. We will explain later in
Section~\ref{section:conclusion} why we choose to not use this
completion.  Intuitively, by adding all colimits we glue the vector
spaces ``too much'', and not only we loose a geometric intuition of
the objects we are dealing with, but we may run into actual technical
problems when it comes the existence of minimal automata.

\AP Instead, we restrict our attention to a class of colimits (which
strictly contains coproducts) for which different spaces in the
colimit can be ``glued'' together along subspaces, but which do not
contain implicit self folding (i.e., such that an element of a vector
space is not glued to a distinct element of the same vector space,
directly or indirectly). E.g., we can describe `two one-dimensional
spaces, the 0-dimensional subspaces of which are identified through a
linear bijection'.  In this way we obtain the new category of ""glued
vector spaces"" and ""hybrid set-vector space automata"",
corresponding to $\MonoDiag\Vec$"-automata@automaton in a category" in
the rest of the paper.

Generic arguments of colimits provide the language for describing
these objects, but do not solve the question of minimality. In
particular, we are interested in automata whose space of
configurations is a \emph{finite} colimit belonging to the class
described above. The categorical development in this work addresses
the minimization problem for hybrid automata.

\subparagraph{An intuition in the case of gluing of vector spaces.}
In the case of gluing of vector spaces, it is possible to isolate a
combinatorial statement that plays a crucial role in the existence of
minimal \kl{hybrid set-vector automata}:
\begin{quote}
  \AP \intro[minimal cover vs]{(a)} Any subset of a finite-dimensional
  vector space admits a minimal cover as a finite union of subspaces.
  \AP \intro[uniqueness cover vs]{(b)} Furthermore, there is a unique
  such cover which is a union of subspaces which are incomparable with
  respect to inclusion.
\end{quote}
\noindent\begin{minipage}{7.3cm}
For instance, in the original \kl{vector space
  automaton}~$\automatonVec$, the states that are reachable in fact
all belong to $\Reals\times\{0\}\cup\{0\}\times\Reals$, and this is
the \kl[minimal cover vs]{minimal cover as in (a)} of these reachable
configurations. This subset of~$\Reals^2$ has the structure of two
$\Reals$-spaces. These happen to intersect at~$(0,0)$, hence it is
necessary to glue them at~$0$ to faithfully represent this set of
reachable configurations. Thanks to \kl[uniqueness cover vs]{(b)} this
decomposition is canonical, and hence can be used for describing the
automaton.
\end{minipage}
~~~
\raisebox{-1.9cm}
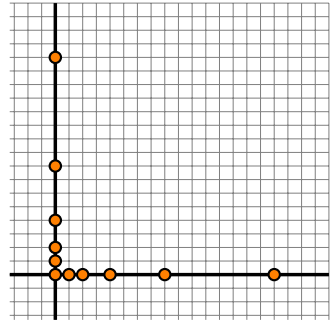
{
\begin{tikzpicture}[scale=0.6,>=stealth', thick, dot/.style = {
    fill=orange, inner sep = 5pt}]
  \def \step {0.3} 

  \draw[help lines,step=\step] (-1,-1) grid (6,6);
  \draw[-,line width=1.3pt] (-1,0) -- (6,0); 
  \draw[-,line width=1.3pt] (0,-1) -- (0,6); 

  \node[draw, circle,inner sep=1.5pt, fill=orange] at (0,0){}; 
  \foreach \s in {0,...,4}{ 
    \node[draw, circle,inner sep=1.5pt, fill=orange] at (2^\s
    *\step,0){}; 
    \node[draw, circle,inner sep=1.5pt, fill=orange] at
    (0,2^\s *\step){}; }
\end{tikzpicture}
}

\section{Automata in a category}
\label{section:categorical automata}
In this section, we provide the general definition for a \kl{(finite
  word) automaton in a category}.
We also isolate properties guaranteeing the existence of minimal
automata.  Though presented differently, the material in the first
subsection is essentially a slight variation around the work of Arbib
and Manes~\cite{ArbibManes75}, which introduced a notion of automaton
in a category and, moreover, highlighted the connection between
factorization systems of the ambient category, duality and
minimization.  In the remaining subsections we develop a refinement of
this approach to minimization, and introduce a notion of factorization
system through a subcategory.

\subsection{\kl{Automata in a category}, \kl{initial automaton}, \kl{final
  automaton}}
\label{subsection:automata}
\begin{definition}[automata]
  \AP Let~$\categoryC$ be a locally small category, $I$ and~$F$ be
  objects of~$\categoryC$, and $\alphabetA$ be some alphabet.  \AP An
  \intro{automaton $\automatonA$ in the category~$\categoryC$} (over
  the alphabet~$\alphabetA$), for short a
  \reintro{$\categoryC,I,F$-automaton} (or simply
  \reintro{$\categoryC$-automaton} when~$I$ and~$F$ are obvious in the
  context), is a tuple $(\stateObject,\initialState, \finalMap,
  \transitionMap)$, where $\intro[\stateObject]{}\stateObject$ is an
  object in $\categoryC$ (called the \intro{state object}),
  $\intro*\initialState\colon I\to \stateObject$ and
  $\intro*\finalMap\colon \stateObject\to F$ are morphisms in
  $\categoryC$ (called \intro[initial morphism]{initial} and
  \intro{final morphisms}), and
  $\transitionMap\colon\alphabetA\to\categoryC(\stateObject,\stateObject)$
  is a function associating to each letter $a\in\alphabetA$ a morphism
  $\delta_a\colon\stateObject\to\stateObject$ in $\categoryC$.  We
  extend the function $\delta$ to $\alphabetA^*$ as with
  $\delta_\epsilon$ being the identity morphism on $\stateObject$ and
  $\delta_{wa}=\delta_a\circ\delta_w$ for all $a\in\alphabetA$ and
  $w\in\alphabetA^*$.

  \AP A \intro{morphism of $\categoryC,I,F$-automata}
  $\morphismAutomataH\colon\automatonA\to\automatonA'$ is a morphism
  $h\colon\stateObject\to\stateObject'$ in $\categoryC$ between the
  \kl{state objects} which commutes with the initial, final and
  transition morphisms:
  \begin{equation}
    \label{eq:1}
    \begin{tikzcd}[column sep={1.5cm,between origins},row sep={0.7cm,between origins}]
      &  \stateObject\ar[dd,"h"] & \stateObject\ar[r,"\delta_a"]\ar[dd,swap,"h"] & \stateObject\ar[dd,"h"] & \stateObject\ar[rd,"f"]\ar[dd,swap,"h"] &\\
      I\ar[ru,"i"]\ar[rd,swap,"i'"] &&&&& F\\
      & \stateObject' & \stateObject'\ar[r,swap,"\delta_a'"] &
      \stateObject' & \stateObject'\ar[ru,swap,"f'"] &
    \end{tikzcd}
  \end{equation}
\end{definition}
\begin{example} The two guiding instantiation of this definition are
  as follows.  When the category~$\categoryC$ is $\Set$, $I=1$
  and~$F=2$, we recover the standard notion of a deterministic and
  complete automaton (over the alphabet $\alphabetA^*$). In the second
  case, when $\categoryC$ is~$\Vec$ over a base field~$\fieldK$,
  $I=\fieldK$ and $F=\fieldK$, we obtain $\fieldK$-weighted
  automata. Indeed, if $\stateObject$ is isomorphic to $\fieldK^n$ for
  some natural number $n$, then linear maps
  $\initialState\colon\fieldK\to\stateObject$ are in one-to-one
  correspondence with vectors $\fieldK^n$. The same holds for linear
  maps $\finalMap\colon\stateObject\to\fieldK$, hence $\initialState$
  and $\finalMap$ are simply selecting an initial, respectively, a
  final vector.
\end{example}

\begin{definition}[languages and language accepted]
  \AP A \intro{$\categoryC,I,F$-language} (or
  \reintro{$\categoryC$-language} when~$I$ and~$F$ are clear from the
  context) is a function~$L\colon\alphabetA^*\to\categoryC(I,F)$.  We
  say that $\automatonA$ \intro{accepts the language~$L$} if
  $L(w)=\intro*\autosem{\automatonA}(w):=\finalMap\circ
  \transitionMap_w\circ\initialState$ for all $w\in\alphabetA^*$.
  \AP Let $\intro[\categoryCautomataL]{}\categoryCautomataL$ denote
  the \intro{category of~$\categoryC,I,F$-automata for $L$}, that is,
  the category whose objects are \kl{$\categoryC,I,F$-automata} that
  \kl{accept the language~$L$} and whose arrows are \kl{morphisms of
    $\categoryC,I,F$-automata}%
  \footnote{\AP If $\automatonA$ \kl{accepts the language $L$} and
    $h\colon\automatonA\to\automatonA'$ is a \kl{morphism of
      $\categoryC,I,F$-automata}, then $\automatonA'$ also \kl{accepts
      the language~$L$}. Hence, $\categoryCautomataL$ is a `connected
    component' in the category of all
    \kl{$\categoryC,I,F$-automata}.}.
\end{definition}

\begin{restatable}{lem}{initFinalAutomaton}%
  \label{lem:initial-and-final-automata}%
  \AP If the coproduct $\coprod_{w\in\alphabet^*}I$ exists
  in~$\categoryC$, then $\categoryCautomataL$ has an \kl{initial
    object}~$\reintro*\initialCautomatonL$.  If the product
  $\prod_{w\in\alphabet^*}F$ exists in~$\categoryC$, then
  $\categoryCautomataL$ has a \kl{final
    object}~$\reintro*\finalCautomatonL$.
\end{restatable}
  \intro[\initialCautomaton]{}%
  \intro[\finalCautomaton]{}%
  \phantomintro{initial automaton}%
  \phantomintro{final automaton}%
  In the case of $\categorySet$, these automata are well known. The
  first one has as states~$\alphabetA^*$, as initial
  state~$\varepsilon$, and when it reads a letter~$a$, its maps $w$
  to~$wa$. Its final map sends the state~$w$ to~$L(w)(0)$.  There
  exists one and exactly one morphism from this automaton to each
  automata for the same language.  The generalisation of this
  construction is that the state space is the coproduct of~$A^*$-many
  copies of~$I$.
  The final automaton is known as the automaton of `residuals'.  Its
  set of states are the maps from~$\alphabetA^*$ to~$2$. The initial
  state is~$L$ itself, and when reading a letter~$a$, the state $S$ is
  mapped to $w\mapsto S(aw)$. The final map sends~$S$
  to~$S(\varepsilon)$. The generalisation of this construction is that
  the state space is the product of~$A^*$-many copies of~$F$.

\subsection{Factorizations through a subcategory}
\label{subsection:factorization through}

It is important in the development of this paper to distinguish the
category $\categoryCautomataL$ in which the \kl{initial and final
  automata for a language}~$L$ exist (recall
Lemma~\ref{lem:initial-and-final-automata}) and which contains
`infinite automata', from the subcategory, named~$\categorySautomataL$
that is used for the concrete automata (with state object in
$\categoryS$) which are intended to be algorithmically manageable.  In
this section, we provide the concept of \kl{factorizing through a
  subcategory}, which articulates the relation between these two
categories.

\begin{definition}[factorization through a subcategory]
  \AP Assume $\S$ is a subcategory of $\categoryC$.  An arrow $f\colon
  X\to Y$ in $\categoryC$ is called \intro{$\S$-small} if it factors
  through some object $S$ of $\S$, that is, $f$ is the
  composite \begin{tikzcd} X\ar[r, "u"]& S\ar[r, "v"] & Y
  \end{tikzcd}
  for some $u\colon X\to S$ and $v\colon S\to Y$.

  \AP A \intro{factorization system through $\S$ on $\C$} (or simply a
  \intro{factorization system on~$\categoryC$}
  if~$\categoryC=\categoryS$) is a pair $(\EpiS,\MonoS)$ where $\EpiS$
  and $\MonoS$ are classes of arrows in $\categoryC$ so that the
  codomains of all arrows in $\EpiS$, the domains of all arrows in
  $\MonoS$ are in $\S$, and the following conditions hold:
  \begin{enumerate}
  \item $\EpiS$ and $\MonoS$ are closed under composition with
    isomorphisms in $\S$, on the right, respectively left
    side.
  \item All \kl{$\S$-small arrows} in $\categoryC$ have an
    \intro{$(\EpiS,\MonoS)$-factorization}, that is, if $f\colon X\to
    Y$ factors through an object of $\S$, then there exists $e\in
    \EpiS$ and $m\in \MonoS$, such that $f=m\circ e$.
  \item The \intro{unique $(\EpiS,\MonoS)$-diagonalization property}
    holds: for each commutative diagram
    \begin{equation}
      \label{eq:9}
      \begin{tikzcd}
        X\ar[r,"e"]\ar[d,swap, "f"] & T\ar[d, "g"]\ar[ld, dashed, "u"] \\
        S\ar[r, swap,"m"] & Y
      \end{tikzcd}
    \end{equation}
    with $e\in \EpiS$ and $m\in \MonoS$, there exists a unique
    \intro{diagonal}, that is, a unique morphism $u\colon T\to S$ such
    that $u\circ e=f$ and $m\circ u=g$.
  \end{enumerate}
\end{definition}
Using standard techniques, we can prove that whenever $(\EpiS,\MonoS)$
is a \kl{factorization system through~$\S$} on $\categoryC$, both
classes $\EpiS$ and $\MonoS$ are closed under composition, their
intersection consists of precisely the isomorphisms in $\categoryS$,
and, as expected, that $(\EpiS,\MonoS)$-factorizations of
$\categoryS$-small morphisms are unique up to isomorphism.

\begin{example}
  Instantiating $(\categoryC,\categoryS)$ to be $(\Set,\SetFin)$
  yields a natural \kl{factorization system through~$\categorySetFin$}
  on~$\Set$ (as the restriction of the standard
  (epi,mono)-factorization system on~$\Set$ to~\kl{$\SetFin$-small
    morphisms}, i.e., the maps of finite image).  Over these
  categories~\kl{$\categorySetFin$-automata} are \kl{deterministic
    finite state automata} inside the more general category of
  \kl{$\categorySet$-automata} which are deterministic (potentially
  infinite) automata.  The example~$(\Vec,\VecFin)$ was already
  mentioned. In this case, being~\kl{$\categoryVecFin$-small} is
  equivalent to having finite rank.

  Notice that for $(\categoryC,\categoryS)=(\Set,\SetFin)$ or
  $(\Vec,\VecFin)$, the factorization systems through the
  subcategories are obtained simply by restricting the factorization
  systems on $\Set$, respectively $\Vec$. This is because, in these
  cases $\categoryS$ is closed under quotients and subobjects
  in~$\categoryC$.  The category~$\GlueFin(\VecFin)$ used in this
  paper is closed under quotients, but in general not under subobjects
  (and this is the important reason for this extension of the standard
  notion of factorization). This is also a case in which there is a
  factorization system in the category~$\categoryC$, that coincide
  over~$\categoryS$ with factorizing through~$\categoryS$, but for
  which factorizing in~$\categoryC$ of an~\kl{$\categoryS$-small}
  morphism and factorizing it through~$\categoryS$ yield different
  results.
\end{example}

A factorization system on $\categoryC$ lifts naturally to categories
of $\categoryC$-valued functors.  Automata being very close in
definition to such a functor category, factorization systems also lift
to them.  Lemma~\ref{lem:factorization-though-automata} shows that
this is also the case for \kl{factorization systems
  through~$\categoryS$}, assuming of course that the input and output
objects $I$ and $F$ belong to $\categoryS$.
\begin{lemma}\label{lem:factorization-though-automata}
  Whenever $(\EpiS,\MonoS)$ is a
  \kl[factorization through]{factorization system
    through a category~$\categoryS$} then  $(\EpiAutoS,\MonoAutoS)$ is a
  \kl{factorization system through} $\categorySautomataL$ for the
  category $\categoryCautomataL$, where~$\intro[\EpiAutoS]{}\EpiAutoS$
  (resp. $\intro[\MonoAutoS]{}\MonoAutoS$) contains these
  \kl{$\categoryCautomataL$-morphisms} that belong to~$\EpiS$
  (resp. to $\MonoS$) as $\categoryC$-morphisms.
\end{lemma}

\subsection{Minimization through a subcategory}
\label{subsection:minimization}

In this section, we show how the joint combination of having
\kl[initial automaton]{initial} and \kl{final automata for a
  language}, as given by Lemma~\ref{lem:initial-and-final-automata},
together with a \kl{factorization system through a
  subcategory~$\categoryS$} yields the existence of a minimal
\kl{$\categoryS$-automaton} for \kl{small}
\kl{$\categoryC$-languages}.

\AP We make the \intro[minimal assumptions]{following assumptions}:
$(\EpiS,\MonoS)$ is a \kl{factorization system through~$\categoryS$
  on~$\categoryC$}, and~$L$ is a \kl{$\categoryC$-language}
\kl{accepted} by some \kl{$\categoryS$-automaton} such that there
exist an \kl[initial automaton]{initial $\initialCautomatonL$}
$\categoryC$-automaton and a \kl{final $\categoryC$-automaton}
$\finalCautomatonL$ for~$L$.

\begin{definition}[minimal automaton]
  The \intro{minimal $\categoryC$-automaton for~$L$}, denoted
  $\intro*\minimalSautomatonL$, is the\footnote{It is unique up to
    isomorphism according to the \kl{diagonal property}.}
  \kl{$\categoryS$-automaton for~$L$} obtained by
  \kl{$(\EpiAutoS,\MonoAutoS)$-factorization} of the unique
  $\categorySautomataL$"-small" morphism from~$\initialCautomatonL$
  to~$\finalCautomatonL$.
\end{definition}

\begin{theorem}%[minimization]
  \label{thm:minimization-through}
  For all \kl{$\categoryS$-automata~$\automatonA$ for~$L$} satisfying
  the \kl[minimal assumptions]{above assumptions}, we have
  \begin{align*}
    \minimalSautomatonL \cong
    \obsSautomaton(\reachSautomaton(\automatonA)) \cong
    \reachSautomaton(\obsSautomaton(\automatonA))\,,
  \end{align*}
  in which \AP
  \begin{itemize}
  \item $\intro*\reachSautomaton(\automatonA)$ is the result of
    applying an \kl{$(\EpiAutoS,\MonoAutoS)$-factorization} to the
    unique \kl{$\categorySautomataL$-morphism}
    from~$\initialCautomatonL$ to~$\automatonA$, and
  \item $\intro*\obsSautomaton(\automatonA)$ is the result of applying
    an \kl{$(\EpiAutoS,\MonoAutoS)$-factorization} to the unique
    \kl{$\categorySautomataL$-morphism} from~$\automatonA$
    to~$\finalCautomatonL$.
  \end{itemize}
\end{theorem}
This theorem does not only state the existence of a minimal automaton,
it also makes transparent how to make effective its construction: if
one possesses both an implementation of~$\reachSautomaton$ and
$\obsSautomaton$, then their sequencing minimises an input
automaton. From the above theorem it immediately follows that
$\minimalSautomatonL$ is both an \kl{$\EpiAutoS$-quotient} of a
\kl{$\MonoAutoS$-subobject} of $\automatonA$ and a
\kl{$\MonoAutoS$-subobject} of an \kl{$\EpiAutoS$-quotient} of
$\automatonA$: the minimal automaton divides every other automaton for
the language.
\begin{proof}[Proof idea.]
  The proof is contained in the following commutative diagram, in
  which $\twoheadrightarrow$ denotes
  \kl{$\categoryCautomataL$-morphisms} in~$\EpiAutoS$, and
  $\rightarrowtail$ \kl{$\categoryCautomataL$-morphisms}
  in~$\MonoAutoS$: 
\begin{center}
    \begin{tikzcd}
[column sep={1.5cm,between origins},row
      sep={0.9cm,between origins}]
      &&& \automatonA\ar[rrrd,bend left=10] & &
      \\
      \initialCautomatonL\ar[rrru,bend left=10]\ar[rrrd,bend
      right=10,two heads]\ar[rr,two heads] &&
      \reachSautomatonA\ar[ru,rightarrowtail]\ar[rr,two heads] &&
      \obsSautomaton(\reachSautomaton(\automatonA))\ar[rr,
      rightarrowtail] && \finalCautomatonL
      \\
      &&& \minimalSautomatonL \ar[rrru,bend
      right=10,rightarrowtail]\ar[ru,dashed, no head]& &
    \end{tikzcd}
  \end{center}
%}
That
  $\obsSautomaton(\reachSautomaton(\automatonA))$ is
  an \kl{$(\EpiAutoS,\MonoAutoS)$-factorization} of the unique
  \kl{$\categoryCautomataL$-morphism} from~$\initialCautomatonL$
  to~$\finalCautomatonL$ follows since~$\EpiAutoS$ is closed under
  composition. By the \kl{unique diagonal property}, it is isomorphic
  to~$\minimalSautomatonL$. The case
  $\reachSautomaton(\obsSautomaton(\automatonA))$ is symmetric.
\end{proof}

\iffalse Assume that $L$ is a \kl{$\categoryC$-language} and assume
there exists an automaton $\automatonA$ in $\categorySautomataL$, that
is, having a \kl{state object} in the category $\categoryS$.  Hence
the unique map $\initialCautomatonL\to\finalCautomatonL$ is
\kl{$\categorySautomataL$-small}, and thus admits a unique (up to
isomorphism) $(\EpiAutoS,\MonoAutoS)$-factorization through an
automaton $\minAutomatonL$. We claim that $\minAutomatonL$ is the
minimal automaton in the category $\categorySautomataL$, in the sense,
that $\minAutomatonL$ is the quotient of a subautomaton of any other
$\categoryS$-automaton recognising $L$.

Indeed, given such an automaton $\automatonA$, one can define its
reachable $\S$-sub-automaton $\reachAutomatonA$ using the
$(\EpiAutoS,\MonoAutoS)$"-factorization" of the unique map from the
initial automaton $\initialCautomatonL$ to $\automatonA$, and its
observable quotient, dually, by taking the
$(\EpiAutoS,\MonoAutoS)$-factorization of the unique map to the final
automaton $\finalCautomatonL$.

Computing the reachable part $\reachAutomatonA$ of $\automatonA$, and
then, its observable quotient $\obsReachAutomatonA$, we obtain
another $(\EpiAutoS,\MonoAutoS)$"-factorization" of the unique map
$\initialCautomatonL\to\finalCautomatonL$. Therefore, $\minAutomatonL$
is isomorphic to $\obsReachAutomatonA$. Moreover, effective
minimization of $\categoryS,I,F$-automata can be achieved, by
providing effective methods for computing $\reachAutomaton$,
respectively $\obsAutomaton$.  \fi

\subsection{A special case of factorization through}
\label{section:fact-thorugh}

So far, the description of factorization and minimization of automata
is very generic. Hereafter, the classes $\EpiS$ and~$\MonoS$ are
constructed along a particular principle which we describe now. In the
next sections we will instantiate this construction when $\categoryS$
is the subcategory~$\GlueFin(\VecFin)$ of \kl{glued vector spaces}.

\AP In this section we fix an \kl{$(\Epi,\Mono)$-factorization system}
on~$\categoryC$ and a
subcategory~$\categoryS\hookrightarrow\categoryC$.

\begin{definition}
  \AP An \intro{$\categoryS$-extremal epimorphism}%
  \footnote{Note that \kl{$\categoryS$-extremal epimorphisms} need not
    be epimorphisms in $\categoryC$.
  }  in $\categoryC$ is
  an arrow $e\colon X\to S$ in $\categoryC$, with $S$ an object in
  $\S$, such that if $e=m\circ g$ where $m$ is in $\Mono$ with domain
  in $\S$, then $m$ is an isomorphism. We set $\MonoS$ to be the class
  of arrows in $\Mono$ with domain in $\categoryS$, and $\EpiS$ to be
  the class of \kl{$\categoryS$-extremal epimorphisms}.
\end{definition}

\begin{definition}
  \AP An \intro{$\MonoS$-subobject in $\S$} of an object $X$ of $\C$
  is an equivalence class up to isomorphism of a morphism $S\to X$
  belonging to $\Mono$, where $S$ is an object of $\S$. The
  $\MonoS$-subobject $S\to X$ is called \intro{proper} if it is not an
  isomorphism.
\end{definition}

\begin{restatable}{lem}{constrFactSystThrough}\label{lem:constr-fact-syst-through}
  Assume the following conditions hold:
  \begin{enumerate}
  \item\label{cond:1} all arrows in $\Mono$ are monomorphisms in
    $\categoryC$,
  \item\label{cond:2} $\categoryS$ is closed under $\Epi$-quotients,
    i.e., if $e\colon S\to T$ is in $\Epi$ with $S$ in $\categoryS$,
    then $T$ is isomorphic to an object of $\categoryS$,
  \item\label{cond:3} the intersection of a nonempty set of
    $\MonoS$-subobjects of an object $X$ of $\categoryC$ exists and is
    an $\Mono_S$-subobject of $X$, and,
  \item\label{cond:4} the pullback of an $\MonoS$-subobject $m\colon
    S\to T$ of $T$ along a morphism $T'\to T$ in $\categoryS$ is an
    $\MonoS$-subobject of $T'$.
  \end{enumerate}
  Then $(\EpiS,\MonoS)$ is a \kl{factorization system through
    $\categoryS$} on $\categoryC$.
\end{restatable}
The next lemma ensures that condition~\ref{cond:3} 
of Lemma~\ref{lem:constr-fact-syst-through} can be replaced with the
weaker version involving only binary \kl{intersections} of
\kl{$\MonoS$-subobjects}, provided that
any infinite descending chain of \kl{$\MonoS$-subobjects} eventually
stabilises (of course, up to isomorphism).

\begin{restatable}{lem}{wellFoundedIntersec}
  \label{lem:well-founded-and-interesec}
  Assume that there are no infinite descending chains of \kl{proper}
  \kl{$\MonoS$-subobjects}
  \[
  \begin{tikzcd}[ampersand replacement=\&]
    X \& S_1\ar[l,rightarrowtail] \& S_2\ar[l,rightarrowtail] \&
    \ldots \ar[l,rightarrowtail]
  \end{tikzcd}
  \]
  and furthermore that the intersection of any two
  \kl{$\MonoS$-subobjects} of an object $X$ of $\categoryC$ exists and
  is an \kl{$\MonoS$-subobject} of $X$, then condition~\ref{cond:3} in
  the hypothesis of Lemma~\ref{lem:constr-fact-syst-through} holds.
\end{restatable}

The proof simply uses finite partial intersections in order to create
a strictly descending chain of \kl{$\MonoS$-subobjects}. By
assumption, this construction has to stop, and the last element of the
sequence happens to be the intersection of the entire
family.\footnote{The attentive reader will have recognised in this
  argument part of the reason why every subset of a finite-dimensional
  vector space admits a minimal cover as a finite union of subspaces.}

\section{Gluing of categories}
\label{section:gluing of categories}
We turn now to the central construction of this paper: given a
category~$\categoryC$ and a subcategory~$\categoryS$, we construct a
category~$\Glue(\categoryC)$ of ``gluings of objects in~$\categoryC$''
that has both~$\categoryC$ and~$\GlueFin(\categoryS)$ -- the category
of ``finite gluings of objects in~$\categoryS$'' -- as subcategories.
Under proper assumptions on~$\categoryC$ and~$\categoryS$, the
resulting pair $(\GlueC,\GlueFinS)$ satisfies the assumption required
for constructing minimisable automata for
\kl{$\GlueC$-languages}. Taking~$\categoryC=\Vec$
and~$\categoryS=\VecFin$ we obtain the construction informally
described in Section~\ref{section:informal presentation}.

\AP Throughout this section we assume that $\categoryC$ is equipped
with a $(\Epi,\Mono)$-factorization system consisting of strong
epimorphisms and monomorphisms.

\subsection{The free gluing of a category}

When $\categoryC$ is small, it is well known that the Yoneda embedding
of $\categoryC$ into the category of presheaves over $\categoryC$ is a
free completion of $\categoryC$ under colimits of small diagrams. For
a possibly large category, one has to consider instead the category of
small presheaves, i.e. small colimits of representable ones, see for
example~\cite{DayLack2007}. For our purposes, we found more
illuminating and direct to use a syntactic way of describing the
colimit completion of a category.

\subparagraph{The category of diagrams.}  \AP Assume $\categoryC$ is a
locally small category. The \intro{free colimit completion of $\C$} is
the category $\intro*\DiagC$ whose objects are \intro{diagrams}
$\functorF\colon\D\to\C$ and morphisms between two diagrams
$\functorF\colon\D\to\C$ and $\functorG\colon\E\to\C$ will be given in
Definition~\ref{def:morph-of-diag}.

\AP To this end we define an equivalence relation on arrows from an
arbitrary object~$X$
\begin{wrapfigure}{r}{2.3cm}
  \vspace{-0.29cm}$\begin{tikzcd}[column sep={1.5cm,between
      origins},row sep={0.6cm,between origins}]
    & \functorG e\ar[dd,"\functorG j"] \\
    X\ar[ru,"g"]\ar[rd,swap,"g'"] &\\
    & \functorG e'
  \end{tikzcd}$
\end{wrapfigure}
of $\C$ to the objects in the image of $\functorG$.
Assume $e,e'$ are objects in $\E$. \AP We consider the least
equivalence relation $\intro*\equivG$ which contains all pairs
$(g,g')$, where $g\colon X\to \functorG e$, $g'\colon X\to \functorG
e'$ are such that there exists $j\colon e\to e'$ a map in $\E$ with
$\functorG j\circ g=g'$, i.e., the diagram on the right commutes.

We denote by $\intro[\Ghat]{}\Ghat(X)$ the equivalence classes of the
relation~$\equivG$.

\begin{definition}
  \label{def:morph-of-diag}
  A \intro[diagram morphism]{morphism between diagrams}
  $\functorF\colon\D\to\C$ and $\functorG\colon\E\to\C$ is a map $f$
  which associates to each object $d$ in $\D$ an equivalence class
  $f(d)\in\Ghat(\functorF d)$, such that whenever $u\colon d\to d'$ is
  a morphism in $\D$ and $g\colon \functorF d'\to \functorG e$ is in
  the equivalence class $f(d')$, then $g\circ \functorF u$ is in the
  equivalence class $f(d)$.
\end{definition}

\subparagraph{The subcategory of gluings.}  \AP We are now ready to
define the category~$\GlueC$, which is a restriction of~$\DiagC$ to
"$\Mono$-diagrams", that is, to diagrams that intuitively `do not
quotient'.  Recall that $\categoryC$ has a factorization system
$(\Epi,\Mono)$ in which~$\Epi$ are the strong epimorphisms and~$\Mono$
are the monomorphisms.
\begin{definition}[\reintro{glued category}]\AP
  An \intro{\Mono-cocone} over a \kl{diagram} $\functorF\colon\D\to\C$
  is a cocone $(u_d\colon\functorF d\to X)_{d\in\categoryD}$ such that
  all the \kl{structural components} of the cocone $u_d$ are in
  $\Mono$.  An \intro{$\Mono$-diagram} is a diagram that has an
  \kl{\Mono-cocone}.

  \AP The \reintro{glued category} $\intro*\Glue(\categoryC)$ is the
  subcategory of $\DiagC$ over the "$\Mono$-diagrams"
  $\functorF\colon\D\to\C$.
  Let~$\intro*\GlueFin(\categoryC)$ the subcategory
  of~$\Glue(\categoryC)$ that has as objects the finite diagrams
  of~$\Glue(\categoryC)$.
\end{definition}
Notice  that, if $\functorF$ is such a diagram, then we can show that for
each morphism $v\colon d \to d'$ in $\D$, we have that $\functorF
u\colon \functorF d\to \functorF d'$ is in $\Mono$ (however this is
not a characterisation).  Also, if there exists a universal
cocone for an \kl{\Mono-diagram}, then this cocone is in particular an
\kl{\Mono-cocone}.
\begin{restatable}{lem}{monoDiagReflect}
  \label{lem:monoDiag-reflect}
  If~$\categoryC$ is cocomplete, then $\MonoDiagC$ is a full
  reflective subcategory of $\DiagC$, and hence $\MonoDiagC$ is a
  cocomplete category. If $\categoryC$ is furthermore complete, then
  so is $\MonoDiagC$.
\end{restatable}
In the automata theoretic application we have in mind, we use this
category in order to construct the \kl{initial and final automata for
  a language}.

\subsection{A factorization system through finite gluings}
\label{subsec:fact-thourgh-inf-glue}

The category of most interest for us is the full subcategory
$\MonoDiagFinS$ of $\MonoDiagC$ which consists of the finite
\kl{\Mono-diagrams} over~$\categoryS$.  In this section we construct
in particular, under suitable assumptions, a \kl{factorization system
  through} $\GlueFinS$ on $\MonoDiagC$, making use of
Lemma~\ref{lem:constr-fact-syst-through}.  For~$\categoryS=\VecFin$,
this is the category of `finite gluings of finite vector spaces' that
we longly introduced in Section~\ref{section:informal presentation}.

\AP We define the following classes of morphisms in $\MonoDiagC$.
\begin{itemize}
\item \intro[\EpiGlueC]{}$\EpiMDC$ consists of the \kl[diag
  morphism]{morphisms} $f\colon\functorF\to\functorG$,
  where~$\functorF\colon\categoryD\to\categoryC$ and
  $\functorF\colon\categoryE\to\categoryC$, such that for all $e$ in
  $\categoryE$ there exists a representative $f_d\colon\functorF
  d\to\functorG e'$ in the equivalence class $f(d)$ and a morphism
  $u\colon \functorG e\to \functorG e'$, so that $u\equivf\functorG
  \id_{\functorG e}$ and $u$ factors through the image of $f_d$.
\item \intro[\MonoGlueC]{}$\MonoMDC$ consists of \kl[diag
  morphism]{morphisms} $f\colon\functorF\to\functorG$ such that for
  all morphisms $u\colon X\to \functorF d$ and $v\colon X\to\functorF
  d'$ such that $f_d\circ u\equivG f_{d'}\circ v$ (for $f_d$ and
  $f_{d'}$ in the equivalence classes $f(d)$, respectively $f(d')$),
  we have that $u\equivF v$.
\end{itemize}
One can easily verify that the arrows in $\MonoMDC$ are exactly the
monomorphisms in $\MonoDiagC$.\footnote{As a side remark, we should
  mention that these classes of arrows correspond precisely to the
  natural transformations between the induced presheaves that are
  pointwise injective.}  The next lemma establishes that under mild
conditions on $\categoryC$ we have a "(strong epi, mono) factorization
system" on
$\MonoDiagC$.

\begin{restatable}{lem}{factSystMonoDiagC}
  \label{lem:fact-syst-mono-diag-C}
  Assume $\categoryC$ has intersections. Then $(\EpiMDC,\MonoMDC)$ is
  a (strong epi, mono) \kl{factorization system on $\MonoDiagC$}.
\end{restatable}

\AP In what follows we say that a subcategory $\categoryS$ of
$\categoryC$ is ""well-behaved"" if it satisfies the hypothesis of
Lemmas~\ref{lem:constr-fact-syst-through}
and~\ref{lem:well-founded-and-interesec} with respect to the "(strong
epi, mono) factorization system" on $\categoryC$. (In fact
condition~\ref{cond:3} 
of Lemma~\ref{lem:constr-fact-syst-through} can be replaced by its
binary version.)

\begin{restatable}{thrm}{factSystThroughMonoDiagFinS}
  \label{lem:fact-syst-through-mono-diag-fin-S}
  Assume $\categoryC$ has intersections and pullbacks. If the
  subcategory $\categoryS$ of $\categoryC$ is "well-behaved", then
  $\GlueFinS$ is a "well-behaved" subcategory of $\MonoDiagC$.
\end{restatable}

\begin{proof}[Some ideas about the proof.]
  This result is an application of
  Lemmas~\ref{lem:constr-fact-syst-through}
  and~\ref{lem:well-founded-and-interesec}.  The central combinatorial
  aspect of this statement is that there exists no infinite \kl[proper
  subobject]{strictly} descending chains of \kl{$\MonoMDC$-subobjects}
  in $\GlueFinS$.  For the sake of contradiction, let us consider a
  descending sequence of \kl{diagrams} from~$\GlueFinS$:
  \[
  \begin{tikzcd}
    \functor{F_0} &\functor{F_1}\ar[l,rightarrowtail,"f_1"'] &
    \functor{F_2}\ar[l,rightarrowtail,"f_2"'] & \ldots
    \ar[l,rightarrowtail,"f_3"']
  \end{tikzcd}
  \]
  We have to prove that it is ultimately constant (up to \kl[diagram
  morphism]{isomorphism}). Let the \kl{diagrams} be
  $\functorF_i:\categoryD_i\rightarrow\categoryS$ for all~$i$. 
  The first step is to consider an \emph{aggregation of
    mono-diagrams}, that is, we construct a big "diagram" that
  aggregates all the~$\functorF_i$'s.  At the level of objects this
  diagram contains the disjoint unions of the~$\categoryD_i$'s. We
  call the objects originating from~$\categoryD_i$ of ""rank~$i$@rank
  main"". Secondly, we prove a \emph{global homogeneity} property of
  $\functorF$: given two arrows~$g\colon X\rightarrow \functorF d$,
  $g'\colon X\rightarrow\functorF d'$ with~$d,d'$ at the same
  "rank~$i$@rank main", then $g\equivf{\functorF}g'$ if and only if
  $g\equivf{\functorF_i}g'$.
  Finally we prove the existence of an isomorphism $g_j$
  from~$\functorF_{j-1}$ to~$\functorF_j$ for some $j$ by analysing
  the structure of $\functorF$ and using König's lemma.
\end{proof}

We come back to the leading example of $\GlueFin(\VecFin)$-"automata".
Applying Theorem~\ref{lem:fact-syst-through-mono-diag-fin-S} for
$(\categoryC,\categoryS)=(\Vec,\VecFin)$ we obtain a "factorization
system through" $\GlueFin(\VecFin)$ on $\Glue(\Vec)$. Using
Lemma~\ref{lem:monoDiag-reflect} and
Theorem~\ref{thm:minimization-through} we derive that "hybrid
set-vector automata" are minimisable.
\begin{corollary}
  \label{cor:min-hybrid}
  For any $\Glue(\Vec)$-"language" accepted by some
  $\GlueFin(\VecFin)$-"automaton" there exists a minimal
  $\GlueFin(\VecFin)$-"automaton". In particular, any
  $\Vec$-"language" accepted by a $\VecFin$-automaton has a minimal
  $\GlueFin(\VecFin)$-"automaton".
\end{corollary}

For the language described in Section~\ref{section:informal
  presentation}, and for which the minimal vector automaton has a
two-dimensional state space, we obtain a minimal
$\GlueFin(\VecFin)$-"automaton" obtained by glueing two
one-dimensional spaces at $0$. Formally, this is an "$\Mono$-diagram"
$\functorF\colon\categoryD\to\VecFin$ where $\categoryD$ is a three
object poset $\{\bot,0,1\}$ with $\bot\le 0$ and $\bot\le 1$. The
functor $\functorF$ maps $0,1$ to one-dimensional spaces, $\bot$ to
the zero-dimensional space and the morphisms $\bot\le 0$ and $\bot\le
1$ to its inclusions in the respective one-dimensional spaces.

For another example, consider the "language" which to a word~$u\in
a^*$ associates the value $\cos(\alpha|u|)$ for some~$\alpha$ which is
not a rational multiple of~$\pi$, and whose minimal "vector space
automaton" has a two-dimensional state space. If we used the
factorization in $\Glue(\Vec)$ we would obtain a glueing of infinitely
many one-dimensional subspaces (obtained by rotations with angle
$\alpha$). Thus, it is crucial for our setting to use the
\emph{"factorization system through"} $\GlueFin(\VecFin)$. In this
case, the minimal $\GlueFin(\VecFin)$-"automaton" also has just a
two-dimensional vector space of states.

\section{Conclusion}
\label{section:conclusion}
We have introduced a new way to construct automata which, thanks to
category-theoretic insights, admits minimal automata `by design'.  The
introductory example of \kl{hybrid set-vector automata} is a
convincing instance of this approach, which has both algorithmic
merits (in succinctness of the encoding of the state space) and
theoretical merits (in that there exists minimal automata). The
closest work to our knowledge is the work of Lombardy and Sakarovitch
\cite{LombardyS06} which studies the sequentialisability of weighted
automata; in the framework of this paper, this is answering the
question whether a "vector space automaton" is equivalent to a "hybrid
set-vector automaton" for which the state space consists of dimension
1 vector spaces only, glued at 0 (the problem remains open).

At the categorical level, we should say a few words regarding our
design choices. First why not use more familiar co-completions such as
the Ind-completion of the free co-completion? The answer is that if we
did so, we would not obtain the desired behaviour when we restrict our
attention to the `finite' automata. For example finite filtered
colimits are not very interesting, while the freely added finite
colimits of vector spaces are not closed under quotients in the free
co-completion, thus the work in the previous sections cannot be
applied.

Another question one may ask is why we haven't used coalgebras, as
in~\cite{AdamekBHKMS:2012} or in the work
of~\cite{AdamekEtAl:Well-Pointed-CoAlg} on well-pointed coalgebras.
First, the factorization through a subcategory, which plays a crucial
role in our work, is not developed in that setting. Secondly, we
believe that the functorial approach to automata, which neatly
combines the dual narrative of automata seen as both algebras and
coalgebras is worth saying. As we show
in~\cite{ColcombetPetrisanCALCO17}, we can employ this framework for
minimizing subsequential transducers \` a la
Choffrut~\cite{Choffrut79} (by interpreting them as automata in a
Kleisli category). This is also an example in which the conditions in
Lemma~\ref{lem:initial-and-final-automata} are not necessary. We
believe, that at least in that situation the functorial approach works
slightly smoother than the coalgebraic one~\cite{Hansen10}. Also, by
changing the input category, we can further extend this work to
capture tree automata or algebras (for instance monoids).

In the particular model of "hybrid set-vector automata" the problem
of effectiveness remains: we have proved the existence of a minimal
automaton for a language, but obtaining the reachable configurations
in an effective way is the subject of ongoing work.

\bibliography{refs}
\end{document}